\def\year{2019}\relax
\documentclass[letterpaper]{article} 
\usepackage{aaai19}  
\usepackage{times}  
\usepackage{helvet}  
\usepackage{courier}  
\usepackage{url}  
\usepackage{graphicx}  
\usepackage{mathtools}
\usepackage{graphicx}
\usepackage{epstopdf}
\usepackage{color}
\usepackage{csquotes}

\usepackage{amsmath}
\usepackage{amsfonts}
\usepackage{amssymb}
\usepackage[subnum]{cases}
\usepackage{booktabs} 
\usepackage{graphics}

\usepackage{amssymb}
\usepackage{lipsum,adjustbox}

\usepackage[subnum]{cases}
\usepackage{color,soul}

\usepackage{algorithm,algorithmicx,algpseudocode}

\usepackage{subcaption}

\usepackage{relsize}

\newtheorem{proof}{Proof}

\newtheorem{lemma}{Lemma}

\usepackage{footnote}
\makesavenoteenv{tabular}

\usepackage{multirow}

\usepackage{empheq}

\begin{document}
%
\title{Outlier Aware Network Embedding for Attributed Networks}

\author{Sambaran Bandyopadhyay \thanks{Also affiliated with Indian Institute of Science, Bangalore}\\
IBM Research, India\\
sambaran.ban89@gmail.com \\
 \And Lokesh N \\ Indian Institute of Science, Bangalore \\ nlokeshcool@gmail.com\\ 
 \And M. N. Murty \\ Indian Institute of Science, Bangalore \\ mnm@iisc.ac.in}

\maketitle

\begin{abstract}
Attributed network embedding has received much interest from the research community as most of the networks come with some content in each node, which is also known as node attributes. Existing attributed network approaches work well when the network is consistent in structure and attributes, and nodes behave as expected. But real world networks often have anomalous nodes. Typically these outliers, being relatively unexplainable, affect the embeddings of other nodes in the network. Thus all the downstream network mining tasks fail miserably in the presence of such outliers. Hence an integrated approach to detect anomalies and reduce their overall effect on the network embedding is required.

Towards this end, we propose an unsupervised outlier aware network embedding algorithm (\textit{ONE}) for attributed networks, which minimizes the effect of the outlier nodes, and hence generates robust network embeddings. We align and jointly optimize the loss functions coming from structure and attributes of the network. To the best of our knowledge, this is the first generic network embedding approach which incorporates the effect of outliers for an attributed network without any supervision. We experimented on publicly available real networks and manually planted different types of outliers to check the performance of the proposed algorithm. Results demonstrate the superiority of our approach to detect the network outliers compared to the state-of-the-art approaches. We also consider different downstream machine learning applications on networks to show the efficiency of ONE as a generic network embedding technique. The source code is made available at \url{https://github.com/sambaranban/ONE}.
\end{abstract}

\section{Introduction}
Network embedding (a.k.a. network representation learning) has gained a tremendous amount of interest among the researchers in the last few years \cite{perozzi2014deepwalk,grover2016node2vec}. Most of the real life networks have some extra information within each node. For example, users in social networks such as Facebook have texts, images and other types of content. Research papers (nodes) in a citation network have scientific content in it. Typically this type of extra information is captured using attribute vectors associated with each node.
The attributes and the link structure of the network are highly correlated according to the sociological theories like homophily \cite{mcpherson2001birds}.
But embedding attributed networks is challenging as combining attributes to generate node embeddings is not easy. Towards this end, different attributed network representation techniques such as \cite{yang2015network,huang2017accelerated,gao2018deep} have been proposed in literature. They perform reasonably well when the network is consistent in its structure and content, and nodes behave as expected.

Unfortunately real world networks are noisy and there are different outliers which even affect the embeddings of normal nodes \cite{liu2017accelerated}. For example, there can be research papers in a citation network with few spurious references (i.e., edges) which do not comply with the content of the papers. There are celebrities in social networks who are connected to too many other users, and generally properties like homophily are not applicable to this type of relationships. So they can also act like potential outliers in the system. Normal nodes are consistent in their respective communities both in terms of link structure and attributes. We categorize outliers in an attributed network into three categories and explain them as shown in Figure \ref{fig:outliers}.

One way to detect outliers in the network is to use some network embedding approach and then use algorithms like isolation forest \cite{liu2008isolation} on the generated embeddings. But this type of decoupled approach is not optimal as outliers adversely affect the embeddings of the normal nodes. So an integrated approach to detect outliers and minimize their effect while generating the network embedding is needed. Recently \cite{liang2018semi} proposes a semi supervised approach for detecting outliers while generating network embedding for an attributed network. But in principle, it needs some supervision to work efficiently. For real world networks, it is difficult to get such supervision or node labels. So there is a need to develop a completely unsupervised integrated approach for graph embedding and outlier detection which can be applicable to any attributed network.

\begin{figure}
\centering
\begin{subfigure}{0.31\linewidth}
\centering
\includegraphics[width=\linewidth]{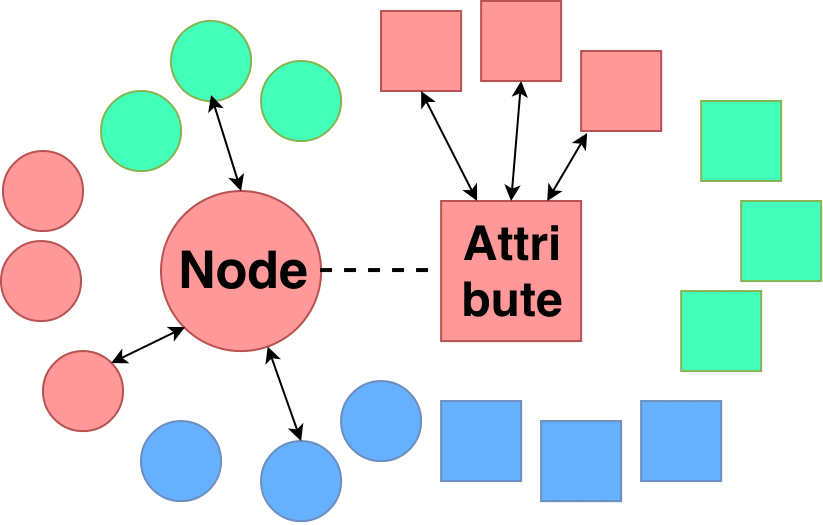}
\caption{}
\end{subfigure}
\begin{subfigure}{0.31\linewidth}
\centering
\includegraphics[width=\linewidth]{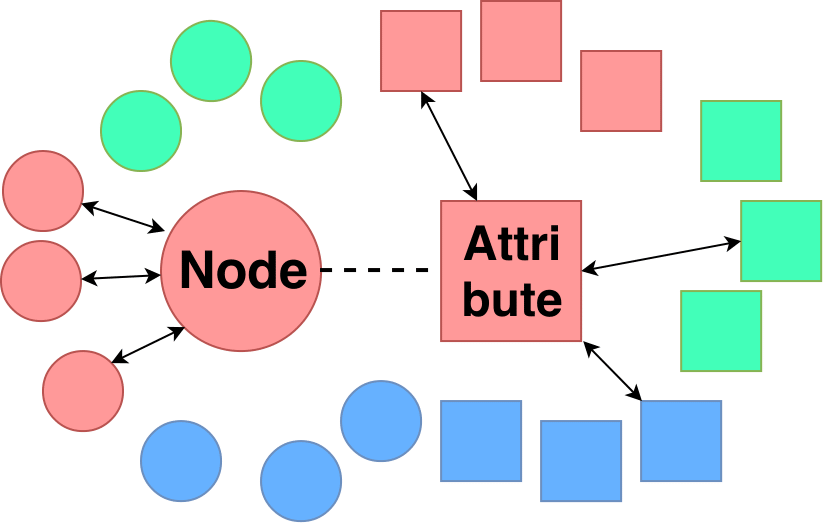}
\caption{}
\end{subfigure}
\begin{subfigure}{0.31\columnwidth}
\centering
\includegraphics[width=\linewidth]{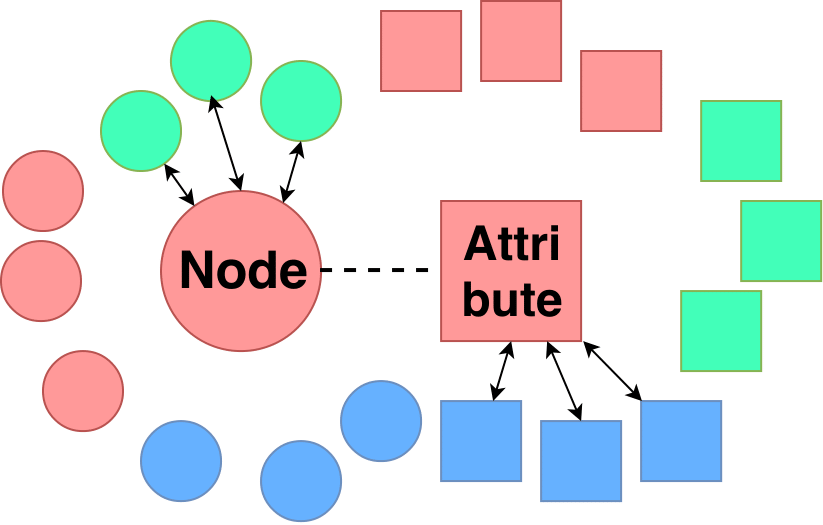}
\caption{}
\end{subfigure}
\caption{This shows different types of outliers that we consider in an attributed network. We highlight the outlier node and its associated attribute by larger circle and rectangle respectively. Different colors represent different communities. Arrows between the two nodes represent network edges and arrows between two attributes represent similarity (in some metric) between them. (a) \textbf{Structural Outlier}: The node has edges to nodes from different communities, i.e., its structural neighborhood is inconsistent. (b) \textbf{Attribute Outlier}: The attributes of the node is similar to attributes of the nodes from different communities, i.e., its attribute neighborhood is inconsistent. (c) \textbf{Combined Outlier}: Node belongs to a community structurally but it has a different community in terms of attribute similarity.}
\label{fig:outliers}
\end{figure}

\textbf{Contributions}: Following are the contributions we make.
\begin{itemize}
\item We propose an unsupervised algorithm called ONE (\textbf{O}utlier aware \textbf{N}etwork \textbf{E}mbedding) for attributed networks. It is an iterative approach to find lower dimensional compact vector representations of the nodes, such that the outliers contribute less to the overall cost function.
\item This is the first work to propose a completely unsupervised algorithm for attributed network embedding integrated with outlier detection. Also we propose a novel method to combine structure and attributes efficiently.
\item We conduct a thorough experimentation on the outlier seeded versions of popularly used and publicly available network datasets to show the efficiency of our approach to detect outliers. At the same time by comparing with the state-of-the-art network embedding algorithms, we demonstrate the power of ONE as a generic embedding method which can work with different downstream machine learning tasks such as node clustering and node classification.
\end{itemize}


\section{Related Work}\label{sec:rw}
This section briefs the existing literature on attributed network embedding, and some outlier detection techniques in the context of networks. Network embedding has been a hot research topic in the last few years and a detailed survey can be found in \cite{hamilton2017representation}. 
Word embedding in natural language processing literature, such as \cite{mikolov2013efficient} inspired the development of node embedding in network analysis.
DeepWalk \cite{perozzi2014deepwalk}, node2vec \cite{grover2016node2vec} and Line \cite{tang2015line} gained popularity for network representation just by using the link structure of the network. DeepWalk and node2vec use random walk on the network to generate node sequences and feed them to language models to get the embedding of the nodes. In Line, two different objective functions have been used to capture the first and second order proximities respectively and an edge sampling strategy is proposed to solve the joint optimization for node embedding. In \cite{ribeiro2017struc2vec}, authors propose struc2vec where nodes having similar substructure are close in their embeddings.

All the papers citeed above only consider link structure of the network for generating embeddings. Research has been conducted on attributed network representation also. TADW \cite{yang2015network} is arguably the first attempt to successfully use text associated with nodes in the network embedding via joint matrix factorization. But their framework directly learns one embedding from content and structure together. In case when there is noise or outliers in structure or content, such a direct approach is prone to be affected more. Another attributed network embedding technique (AANE) is proposed in \cite{huang2017accelerated}. The authors have used symmetric matrix factorization to get embeddings from the similarity matrix over the attributes, and use link structure of the network to ensure that the embeddings of the two connected nodes are similar. A semi-supervised attributed embedding is proposed in \cite{huang2017label} where the label information of some nodes are used along with structure and attributes. The idea of using convolutional neural networks for graph embedding has been proposed in \cite{niepert2016learning,kipf2016semi}. An extension of GCN with node attributes (GraphSage) has been proposed in \cite{hamilton2017inductive} with an inductive learning setup. These methods do not manage outliers directly, and hence are often prone to be affected heavily by them.

Recently a semi-supervised deep learning based approach SEANO \cite{liang2018semi} has been proposed for outlier detection and network embedding for attributed networks. For each node, they collect its attribute and the attributes from the neighbors, and smooth out the outliers by predicting the class labels (on the supervised set) and node context. But getting labeled nodes for real world network is expensive. So we aim to design an unsupervised attributed network embedding algorithm which can detect and minimize the effect of outliers while generating the node embeddings.  

\section{Problem Formulation} \label{sec:prob}
An information network is typically represented by a graph as $\mathcal{G} = (V, E, C)$, where $V=\{v_1, v_2,\cdots, v_N\}$ is the set of nodes (a.k.a. vertexes), each representing a data object. $E \subset \{(v_i,v_j) | v_i,v_j \in V \}$ is the set of edges between the vertexes. 
Each edge $e \in E$ is an ordered pair $e = (v_i, v_j)$ and is associated with a weight $w_{v_i,v_j} > 0$, which indicates the strength of the relation. If $\mathcal{G}$ is undirected, we have $(v_i, v_j) \equiv (v_j, v_i)$ and $w_{v_i,v_j} \equiv w_{v_j,v_i}$; if $\mathcal{G}$ is unweighted, $w_{v_i,v_j} = 1$, $\forall (v_i,v_j) \in E$.

Let us denote the $N \times N$ dimensional adjacency matrix of the graph $\mathcal{G}$ by $A = (a_{i,j})$, where $a_{i,j}=w_{v_i,v_j}$ if $(v_i,v_j) \in E$, and $a_{i,j}=0$ otherwise. So $i$th row of $A$ contains the immediate neighborhood information for node $i$. Clearly for a large network, the matrix $A$ is highly sparse in nature.
$C$ is a $N \times D$ matrix with $C_{i \cdot}$ as rows, where $C_{i \cdot} \in \mathbb{R}^D$ is the attribute vector associated with the node $v_i \in V$. $C_{id}$ is the value of the attribute $d$ for the node $v_i$. For example, if there is only textual content in each node, $c_i$ can be the tf-idf vector for the content of the node $v_i$.

Our goal is to find a low dimensional representation of $\mathcal{G}$ which is consistent with both the structure of the network and the content of the nodes. More formally, for a given network $\mathcal{G}$, network embedding is a technique to learn a function $f : v_i \mapsto \mathbf{y_i} \in \mathbb{R}^K$, i.e., it maps every vertex to a $K$ dimensional vector, where $K < min(N,D)$. The representations should preserve the underlying semantics of the network. Hence the nodes which are close to each other in terms of their topographical distance or similarity in attributes should have similar representations. We also need to reduce the effect of outliers, so that the representations for the other nodes in the network are robust.

\section{Solution Approach: ONE}
We describe the whole algorithm in different parts.
\subsection{Learning from the Link Structure}
Given graph $\mathcal{G}$, each node $v_i$ by default can be represented by the $i$th row $A_{i\cdot}$ of the adjacency matrix. Let us assume the matrix $G$ $\in \mathbb{R}^{N \times K}$ be the network embedding of $\mathcal{G}$, only by considering the link structure. Hence row vector $G_{i\cdot}$ is the $K$ dimensional ($K < min(N,D)$) compact vector representation of node $v_i$, $\forall v_i \in V$. Also let us introduce a $K \times N$ matrix $H$ to minimize the reconstruction loss: $\sum\limits_{i=1}^N\sum\limits_{j=1}^N (A_{ij} - G_{i\cdot} \cdot H_{\cdot j})^2$, where $H_{\cdot j}$ is the $j$th column of $H$, and $G_{i\cdot} \cdot H_{\cdot j}$ is the dot product between these two vectors\footnote[1]{We treat both row vector and column vector as vectors of same dimension, and hence use the dot product instead of transpose to avoid  cluttering of notation}. $k$th row of $H$ can be interpreted as the $N$ dimensional description of $k$th feature, where $k = 1,2,\cdots,K$. This reconstruction loss tends to preserve the original distances in the lower dimensional spaces as shown by \cite{cunningham2015}. But if the graph has anomalous nodes, they generally affect the embedding of the other (normal) nodes. To minimize the effect of such outliers while learning embeddings from the structure, we introduce the structural outlier score $O_{1i}$ for node $v_i \in V$, where $0 < O_{1i} \leq 1$. The bigger the value of $O_{1i}$, the more likely it is that node $v_i$ is an outlier, and lesser should be its contribution to the total loss. Hence we seek to minimize the following cost function w.r.t. the variables $O_1$ (set of all structural outlier scores), $G$ and $H$.
\begin{align}\label{eq:L1}
\mathcal{L}_{str} = \sum_{i=1}^N\sum_{j=1}^N\log\Big({\frac{1}{O_{1i}}}\Big) (A_{ij} - G_{i\cdot} \cdot H_{\cdot j})^2
\end{align}
We also assume $\sum\limits_{i=1}^N O_{1i} = \mu$, $\mu$ being the total outlier score of the network. Otherwise minimizing Eq. 
\ref{eq:L1} amounts to assigning $1$ to all the outlier scores, which makes the loss value 0. It can be readily seen that, when $O_{1i}$ is very high (close to 1) for a node $v_i$, the contribution of this node $\sum\limits_{j=1}^N\log\Big({\frac{1}{O_{1i}}}\Big) (A_{ij} - G_{i.} \cdot H_{.j})^2$ becomes negligible, and when $O_{1i}$ is small (close to 0), the corresponding contribution is high. So naturally, the optimization would concentrate more on minimizing the contributions of the outlier (w.r.t. the link structure) nodes, to the overall objective, as desired. 

\subsection{Learning from the Attributes}
Similar to the case of structure, here we try to learn a $K$ dimensional vectorial representation $U_{i\cdot}$ from the given attribute matrix $C$, where $C_{i\cdot}$ is the attribute vector of node $v_i$. Let us consider the matrices $U \in \mathbb{R}^{N \times K}$ and $V \in \mathbb{R}^{K \times D}$, $U$ being the network embedding just respecting the set of attributes. In the absence of outliers, one can just minimize the reconstruction loss $\sum\limits_{i=1}^N\sum\limits_{d=1}^D (C_{id} - U_{i\cdot} \cdot V_{\cdot d})^2$ with respect to the matrices $U$ and $V$. But as mentioned before, outliers even affect the embeddings of the normal nodes. Hence to reduce the effect of outliers while learning from the attributes, we introduce the attribute outlier score $O_{2i}$ for node $v_i \in V$, where $0 < O_{2i} \leq 1$. Larger the value of $O_{2i}$, higher the chance that node $v_i$ is an attribute outlier. Hence we minimize the following cost function w.r.t. the variables $O_2$, $U$ and $V$.
\begin{align}\label{eq:L2}
\mathcal{L}_{attr} = \sum_{i=1}^N\sum_{d=1}^C\log\Big({\frac{1}{O_{2i}}}\Big) (C_{id} - U_{i \cdot} \cdot V_{\cdot d})^2
\end{align}
We again assign the constraint that $\sum\limits_{i=1}^N O_{2i} = \mu$ for the reason mentioned before. Hence contributions from the non-outlier (w.r.t. attributes) nodes would be bigger while minimizing Eq. \ref{eq:L2}.

\subsection{Connecting Structure and Attributes}
So far, we have considered the link structure and the attribute values of the 
network separately. Also the optimization variables of Eq. \ref{eq:L1} and that in Eq. \ref{eq:L2} are completely disjoint. But optimizing them independently is not desirable as, our ultimate goal is to get a joint low dimensional representation of each node in the network. Also we intend to regularize structure with respect to attributes and vice versa. As discussed before, link structure and attributes in a network are highly correlated and they can be often noisy individually.

One can see that, $G_{i \cdot}$ and $U_{i \cdot}$ are the representation of the same node $v_i$ with respect to structure and attributes respectively. So one can easily act as a regularizer of the other. Also as they contribute to the embedding of the same node, it makes sense to minimize $\sum\limits_{i=1}^N \sum\limits_{k=1}^K (G_{ik} - U_{ik})^2$. But it is important to note that, there is no explicit guarantee that the features in $G_{i \cdot}$ and features in $U_{i \cdot}$ are aligned, i.e., $k$th feature of the structure embeddings can be very different from the $k$th feature of attribute embeddings. Hence before minimizing the distance of $G_{i \cdot}$ and $U_{i \cdot}$, it is important to align the features of the two embedding spaces.

\subsubsection{Embedding Transformation and Procrustes problem}
To resolve the issue above, we seek to find a linear map $W \in \mathbb{R^{K \times K}}$ which transforms the features from the attribute embedding space to structure embedding space. More formally we want to find a matrix $W$ which minimizes $||G - WU||_F$. This type of transformation has been used in the NLP literature, particularly for machine translation \cite{lample2018word}.
If we further restrict $W$ to be an orthogonal matrix, then a closed form solution can be obtained from the solution concept of Procrustes problem \cite{schonemann1966generalized} as follows:
\begin{align}\label{eq:W}
W^* 
= \underset{W \in \mathcal{O}_K}{\text{argmin}} ||G - UW^T||_F
\end{align}
where $W^* = XY^T$ with $X \Sigma Y^T = \text{SVD}(G^T U)$, $\mathcal{O}_K$ is the set of all orthogonal matrices of dimension $K \times K$. Restricting $W$ to be an orthogonal matrix has also several other advantages as shown in the NLP literature \cite{xing2015normalized}.

But we cannot directly use the solution of Procrustes problem, as we have anomalies in the network.
As before, we again reduce the effect of anomalies to minimize the disagreement between the structural embeddings and attribute embeddings. Let us introduce the disagreement anomaly score $O_{3i}$ for a node $v_i \in V$, where $0 < O_{3i} \leq 1$. Disagreement anomalies are required to manage the anomalous nodes which are not anomalies in either of structure or attributes individually, but they are inconsistent when considering them together. Following is the cost function we minimize.
\begin{align}\label{eq:L3}
\mathcal{L}_{dis} = \sum\limits_{i=1}^N \sum\limits_{k=1}^K \log\Big({\frac{1}{O_{3i}}}\Big) \Big(G_{ik} - U_{i\cdot} \cdot (W^T)_{\cdot k}\Big)^2
\end{align}
$\sum\limits_{i=1}^N O_{3i}= \mu$. We will use the solution of Procrustes problem after applying a simple trick to  the cost function above, as shown in the derivation of the update rule of $W$ later.

\subsection{Joint Loss Function}
Here we combine the three cost functions mentioned before, and minimize the following with respect to $G$, $H$, $U$, $V$, $W$ and $O$ ($O$ contains all the variables from $O_1$, $O_2$ and $O_3$).
\begin{align}\label{eq:L}
\mathcal{L} = \mathcal{L}_{str} + \alpha \mathcal{L}_{attr} + \beta \mathcal{L}_{dis}
\end{align}
The full set of constrains are:
\[
0 < O_{1i}, O_{2i}, O_{3i} \leq 1 \; , \; \forall v_i \in V
\]
\[
\sum_{i=1}^{N}O_{1i} = \sum_{i=1}^{N}O_{2i} = \sum_{i=1}^{N}O_{3i} = \mu
\]
\[
W \in \mathcal{O}_{K} \iff W^{T}W = \mathcal{I}
\]


Here $\alpha, \beta > 0$ are weight factors. We will discuss a method to set them in the experimental evaluation. It is to be noted that, the three anomaly scores $O_{1i}$, $O_{2i}$ and $O_{3i}$ for any node $v_i$ are actually connected by the cost function \ref{eq:L}. For example, if a node is anomalous in structure, $O_{1i}$ would be high and its embedding $G_{i\cdot}$ may not be optimized well. So this in turn affects its matching with transformed $U_{i\cdot}$, and hence $O_{3i}$ would be given a higher value to minimize the disagreement loss.

\subsection{Derivations of the Update Rules}
We will derive the necessary update rules which can be used iteratively to minimize Eq. \ref{eq:L}. 
We use the alternating minimization technique, where we derive the update rule for one variable at a time, keeping all others fixed.

\subsection{Updating $G$, $H$, $U$, $V$}
We need to take the partial derivative of $\mathcal{L}$ (Eq. \ref{eq:L}) w.r.t one variable at a time and equate that to zero. For example, $\frac{\partial \mathcal{L}}{\partial G_{ik}} = 0 \Rightarrow \sum\limits_{j=1}^N \log \big(\frac{1}{O_{1i}}\big) (A_{ij} - G_{i\cdot} \cdot H_{\cdot j})(-H_{kj})
+ \log \big(\frac{1}{O_{3i}}\big) (G_{ik} - U_{i\cdot}\cdot (W^T)_{\cdot k}) = 0$.
Solving it for $G_{ik}$,
\begin{empheq}[box=\fbox]{align}\label{eq:G}
\fontsize{7.0pt}{7.5pt} \selectfont
G_{ik} = \frac{
G_{ik}^{num1} + \beta \log\Big({\frac{1}{O_{3i}}}\Big)(W_{k \cdot}\cdot U_{i\cdot})
}
{
\log\Big({\frac{1}{O_{1i}}}\Big)(H_{k \cdot}\cdot H_{k \cdot}) + \beta \log\Big({\frac{1}{O_{3i}}}\Big)
} \\
G_{ik}^{num1} = \log\Big({\frac{1}{O_{1i}}}\Big)\sum_{j=1}^N(A_{ij} - \sum_{k^{'} \neq k}G_{ik^{'}}H_{k^{'}j})H_{kj}
\nonumber
\end{empheq}
\vskip-22mm
Similarly we can get the following update rules.
\begin{align}\label{eq:H}
 \boxed{
H_{kj} = \frac{
\sum_{i=1}^N\log\Big({\frac{1}{O_{1i}}}\Big)(A_{ij} - \sum_{k^{'} \neq k}G_{ik^{'}}H_{k^{'}j})G_{ik}
}
{
\sum_{i=1}^N\log\Big({\frac{1}{O_{1i}}}\Big)G_{ik}^{2}
}
 }
\end{align}
\begin{empheq}[box=\fbox]{align}\label{eq:U}
\tiny
U_{ik} = \frac{
U_{ik}^{num1} + U_{ik}^{num2}
}
{
\beta\log\Big({\frac{1}{O_{2i}}}\Big)(V_{k\cdot}\cdot V_{k\cdot}) + \gamma\log\Big({\frac{1}{O_{3i}}}\Big)W_{\cdot k}\cdot W_{\cdot k}
} \\
U_{ik}^{num1} = \beta\log\Big({\frac{1}{O_{2i}}}\Big)\sum_{d=1}^D(C_{id} - \sum_{k^{'} \neq k}U_{ik^{'}}V_{k^{'}d})V_{kd} \nonumber \\
U_{ik}^{num2} = \gamma\log\Big({\frac{1}{O_{3i}}}\Big)\Big(G_{i\cdot} - (U_{i\cdot}W) - (U_{ik}\cdot W_{\cdot k})\Big) \nonumber
\end{empheq}
\vskip-28mm
\begin{align}\label{eq:V}
\boxed{
V_{kd} = \frac{
\sum_{i=1}^N\log\Big({\frac{1}{O_{2i}}}\Big)(C_{id} - \sum_{k^{'} \neq k}U_{ik^{'}}V_{k^{'}d})U_{ik}
}
{
\sum_{i=1}^N\log\Big({\frac{1}{O_{2i}}}\Big)U_{ik}^{2}
}
}
\end{align}

\subsection{Updating $W$}
We use a small trick to directly apply the closed form solution of Procrustes problem as follows.
\begin{align}
\mathcal{L}_{dis} &= \sum\limits_{i=1}^N \sum\limits_{k=1}^K \log\Big({\frac{1}{O_{3i}}}\Big) \Big(G_{ik} - U_{i\cdot} \cdot (W^T)_{\cdot k}\Big)^2 \nonumber \\
&= \sum\limits_{i=1}^N \sum\limits_{k=1}^K \Big(\bar{G}_{ik} - \bar{U}_{i\cdot} \cdot (W^T)_{\cdot k}\Big)^2
\end{align}
Here the new matrices are defined as, $(\bar{G})_{i,k} = \sqrt[]{\log\Big({\frac{1}{O_{3i}}}\Big)} G_{ik}$ and $\bar{U}_{ik} = \sqrt[]{\log\Big({\frac{1}{O_{3i}}}\Big)} U_{ik}$. Say, $X \Sigma Y^T = \text{SVD}(\bar{G}^T \bar{U})$, then $W$ can be obtained as:
\begin{align}\label{eq:W}
\boxed{
W = XY^T
}
\end{align}

\subsection{Updating $O$}
We derive the update rule for $O_1$ first. Taking the Lagrangian of Eq. \ref{eq:L1} with respect to the constraint $\sum\limits_{i=1}^N O_{1i} = \mu$, we get,
\begin{align}
& \frac{\partial }{\partial O_{1i}} \sum\limits_{i,j}\log\Big({\frac{1}{O_{1i}}}\Big) (A_{ij} - G_{i\cdot} \cdot H_{\cdot j})^2 + \lambda(\sum\limits_{i} O_{1i} - \mu) \nonumber  
\end{align}
$\lambda \in \mathbb{R}$ is the Lagrangian constant. Equating the partial derivative w.r.t. $O_{1i}$ to 0:
\begin{align*}
&-\frac{(A_{ij} - G_{i\cdot} \cdot H_{\cdot j})^2}{O_{1i}} + \lambda = 0, \;
\Rightarrow O_{1i} = \frac{(A_{ij} - G_{i\cdot} \cdot H_{\cdot j})^2}{\lambda}
\end{align*}
So, $\sum\limits_{i=1}^N O_{1i} = \mu$ implies $\sum\limits_{i=1}^N \frac{(A_{ij} - G_{i\cdot} \cdot H_{\cdot j})^2}{\lambda} = \mu$. Hence,
\begin{align}\label{eq:O1}
\boxed{
O_{1i} = \frac{
\Big(\sum_{j=1}^N(A_{ij} - G_{i\cdot}\cdot H_{\cdot j})^2 \Big) \cdot \mu
}
{
\sum_{i=1}^N\sum_{j=1}^N(A_{ij} - G_{i\cdot}\cdot H_{\cdot j})^{2}
}
}
\end{align}
It is to be noted that, if we set $\mu = 1$, the constraints $0 < O_{i1} \leq 1$, $\forall v_i \in V$, are automatically satisfied. Even it is possible to increase the value of $\mu$ by a trick similar to \cite{gupta2012integrating}, but experimentally we have not seen any advantage in increasing the value of $\mu$. Hence, we set $\mu = 1$ for all the reported experiments. A similar procedure can be followed to derive the update rules for $O_2$ and $O_3$.
\begin{align}\label{eq:O2}
\boxed{
O_{2i} = \frac{
\Big(\sum_{d=1}^D(C_{id} - U_{i\cdot}\cdot V_{\cdot d})^2 \Big) \cdot \mu
}
{
\sum_{i=1}^N\sum_{j=1}^D(C_{id} - U_{i\cdot}\cdot V_{\cdot d})^{2}
}
}
\end{align}
\begin{align}\label{eq:O3}
\boxed{
O_{3i} = \frac{
\Big( \sum_{k=1}^K(G_{ik} - W_{i\cdot}\cdot U_{\cdot k})^2 \Big) \cdot \mu
}
{
\sum_{i=1}^N\sum_{k=1}^K(G_{ik} - W_{i\cdot}\cdot U_{\cdot k})^2
}
}
\end{align}

\subsection{Algorithm: ONE}\label{sec:algOne}
With the update rules derived above, we summarize ONE in Algorithm \ref{alg:oane}. variables $G$, $H$, $U$ and $V$ can be initialized by any standard matrix factorization technique ($A \approx GH$ and $C \approx UV$) such as \cite{lee2001algorithms}. As our algorithm is completely unsupervised, we assume not to have any prior information about the outliers. So initially we set equal outlier scores to all the nodes in the network and normalize them accordingly. At the end of the algorithm, one can take the final embedding of a node as the average of the structural and the transformed attribute embeddings. Similarly the final outlier score of a node can be obtained as the weighted average of three outlier scores. 

\begin{lemma}
The joint cost function in Eq. \ref{eq:L} decreases after each iteration (steps 4 to 6) of the for loop of Algorithm \ref{alg:oane}.
\end{lemma}
\begin{proof}
It is easy to check that the joint loss function $\mathcal{L}$ is convex in each of the variables $G,H,U,V$ and $O$, when all other variables are fixed. Also from the Procrustes solution, update of $W$ also minimizes the cost function. So alternating minimization guarantees decrease of cost function after every update till convergence.
\end{proof}

The computational complexity of each iteration (Steps 4 to 6 in Algo. \ref{alg:oane}) takes $O(N^2)$ time (assuming $K$ is a constant) without using any parallel computation, as updating each variable $G_{ik}, H_{kj}, V_{kd}, O_{1i}, O_{2i}, O_{3i}$ and $W$ takes $O(N)$ time. But we observe that ONE converges very fast on any of the datasets, as updating one variables amounts to reaching the global minima of the corresponding loss function when all other variables are fixed. The run time can be improved significantly by parallelizing the computation as done in \cite{huang2017accelerated}.

\begin{algorithm} 
  \caption{\textbf{ONE}}
  \label{alg:oane}
\begin{algorithmic}[1]
      
	\Statex \textbf{Input}: The network $\mathcal{G}=(V,E,C)$, $K$: Dimension of the embedding space where $K<min(n,d)$, ratio parameter $\theta$
    \Statex \textbf{Output}: The node embeddings of the network $G$, Outlier score of each node $v in V$
	\State Initialize $G$ and $H$ by standard matrix factorization on $A$, and $U$ and $V$ by that on $C$.
    \State Initialize the outlier scores $O_1$, $O_2$ and $O_3$ uniformly.
	\For{until stopping condition satisfied}
    	\State Update $W$ by Eq. \ref{eq:W}.
		\State Update $G$, $H$, $U$ and $V$ by Eq. from \ref{eq:G} to \ref{eq:V}.
        \State Update outlier scores by Eq. from \ref{eq:O1} to \ref{eq:O3}.
    \EndFor
    \State Embedding for the node $v_i$ is $\frac{G_{i\cdot} + U_{i\cdot} (W^T)}{2}$, $\forall v_i \in V$.
    \State Final outlier score for the node $v_i$ is a weighted average of $O_{1i}$, $O_{2i}$ and $O_{3i}$, $\forall v_i \in V$.
	\end{algorithmic}
  \end{algorithm} 


\section{Experimental Evaluation}
In this section, we evaluate the performance of the proposed algorithm on multiple attributed network datasets and compare the results with several state-of-the-art algorithms.

\subsection{Datasets Used and Seeding Outliers}
To the best of our knowledge, there is no publicly available attributed networks with ground truth outliers available. So we take four publicly available attributed networks with ground truth community membership information available for each node. The datasets are WebKB, Cora, Citeseer and Pubmed\footnote[2]{Datasets: \url{https://linqs.soe.ucsc.edu/data}}. 
To check the performance of the algorithms in the presence of outliers, we manually planted a total of 5\% outliers (with equal numbers for each type as shown in Figure \ref{fig:outliers}) in each dataset. The seeding process involves: (1) computing the probability distribution of number of nodes in each class, (2) selecting a class using these probabilities. For a structural outlier: (3) plant an outlier node in the selected class such that the node has ($m$ $\pm$ 10\%) of edges connecting nodes from the remaining (unselected) classes where $m$ is the average degree of a node in the selected class and (4) the content of the structural outlier node is made semantically consistent with the keywords sampled from the nodes of the selected class.
A similar approach is employed for seeding the other two types of outliers.
The statistics of these seeded datasets are given in Table \ref{tab:data}.
Outlier nodes apparently have similar characteristics in terms of degree, number of nonzero attributes, etc., and thus we ensured that they cannot be detected trivially.

\begin{table}
    \caption{Summary of the datasets (after planting outliers).} 
	\centering
    \resizebox{0.8\columnwidth}{!}{%
	\begin{tabular}{*5c}
	\toprule
	\sffamily{Dataset} & \#Nodes & \#Edges & \#Labels & \#Attributes \\
    \hline
	\midrule
    \sffamily{WebKB} & 919 & 1662  & 5 & 1703 \\
    \sffamily{Cora} & 2843 & 6269 & 7 & 1433 \\
    \sffamily{Citeseer} & 3477 & 5319 & 6 & 3703 \\
    \sffamily{Pubmed} & 20701 & 49523 & 3 & 500\\
\bottomrule
	\end{tabular}
    }
	\label{tab:data}
	\end{table} 

\subsection{Baseline Algorithms and Experimental Setup}
We use DeepWalk \cite{perozzi2014deepwalk}, node2vec \cite{grover2016node2vec}, Line \cite{tang2015line}, TADW \cite{yang2015network}, AANE \cite{huang2017accelerated}, GraphSage \cite{hamilton2017inductive} and SEANO \cite{liang2018semi} as the baseline algorithms to be compared with. The first three algorithms consider only structure of the network, the last four consider both structure and node attributes. We mostly use the default settings of the parameters values in the publicly available implementations of the respective baseline algorithms. As SEANO is semi-supervised, we include 20\% of the data with their true labels in the training set of SEANO to produce the embeddings.

For ONE, we set the values of $\alpha$ and $\beta$ in such a way that three components in the joint losss function in Eq. \ref{eq:L} contribute equally before the first iteration of the for loop in Algorithm \ref{alg:oane}. For all the experiments we keep embedding space dimension to be three times the number of ground truth communities. For each of the datasets, we run the for loop (Steps 4 to 6 in Alg. \ref{alg:oane}) of ONE only 5 times. We observe that ONE converges very fast on all the datasets. Convergence rate has been shown experimentally in Fig. \ref{fig:loss}.
\begin{figure}[h!]
  \centering
  \begin{subfigure}[b]{0.49\linewidth}
    \includegraphics[width=\linewidth]{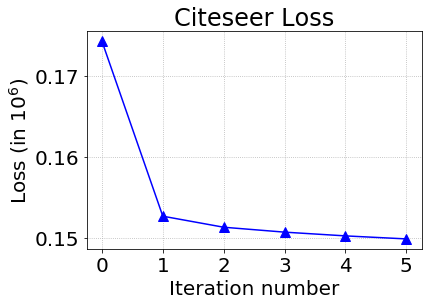}
  \end{subfigure}
  \begin{subfigure}[b]{0.49\linewidth}
    \includegraphics[width=\linewidth]{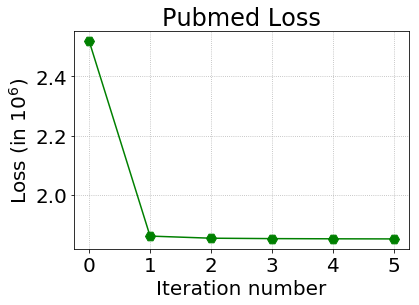}
  \end{subfigure}
   \caption{Values of Loss function over different iterations of ONE for Citeseer and Pubmed (seeded) datasets}
   \label{fig:loss}
\end{figure}

\subsection{Outlier Detection}
A very  important goal of our work is to detect outliers while generating the network embedding. In this section, we see the performance of all the algorithms in detecting outliers that we planted in each dataset.
SEANO and ONE give outlier scores directly as the output. For ONE, we rank the nodes in order of the higher values of a weighted average of three outlier scores (more the value of this average outlier score, more likely the vertex is an outlier). We have observed experimentally that $O_2$ is more important to determine outliers. For SEANO, we rank the nodes in order of the lower values of the weight parameter $\lambda$ (lower the value of $\lambda$ more likely the vertex is an outlier). For other embedding algorithms, as they do not explicitly output any outlier score, we use isolation forest to detect outliers from the node embeddings generated by them.

We use recall to check the performance of each embedding algorithm to find outliers. As the total number of outliers in each dataset is 5\%, we start with the top 5\% of the nodes in the ranked list (L) of outliers, and continue till 25\% of the nodes, and calculate the recall for each set with respect to the artificially planted outliers. Figure \ref{fig:outdetec} shows that ONE, though completely unsupervised in nature, is able to outperform SEANO mostly on all the datasets. 
SEANO considers the role of predicting the class label or context of a node based on only its attributes, or the set of attributes from its neighbors, and accordingly fix the outlierness of the node. Whereas, ONE explicitly compares structure, attribute and their combination to detect outliers and then reduces their effect iteratively in the optimization process. So discriminating outliers from the normal nodes becomes somewhat easier for ONE. As expected, all the other embedding algorithms (by running isolation forest on the embedding generated by them) perform poorly on all the datasets to detect outliers, except on Cora where AANE performs good. 


\begin{figure*}[h!]
  \centering
  \begin{subfigure}[b]{0.21\linewidth}
    \includegraphics[width=\linewidth]{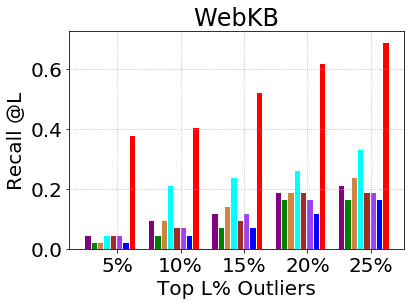}
  \end{subfigure}
  \begin{subfigure}[b]{0.21\linewidth}
    \includegraphics[width=\linewidth]{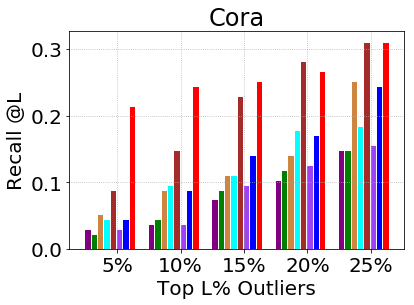}
  \end{subfigure}
  \begin{subfigure}[b]{0.21\linewidth}
    \includegraphics[width=\linewidth]{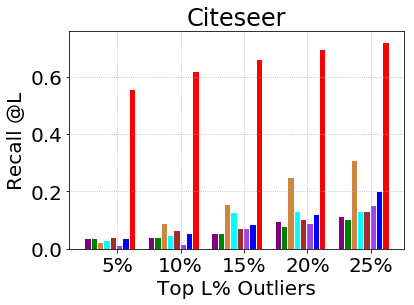}
  \end{subfigure}
  \begin{subfigure}[b]{0.21\linewidth}
    \includegraphics[width=\linewidth]{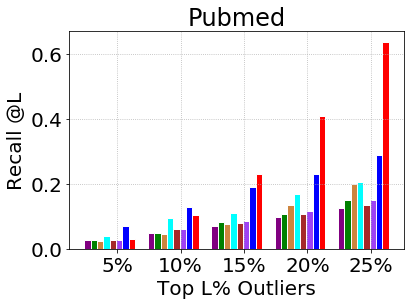}
  \end{subfigure}
  \begin{subfigure}[b]{0.11\linewidth}
    \includegraphics[width=\linewidth]{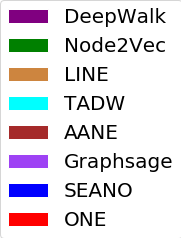}
  \end{subfigure}
  \caption{Outlier Recall at top L\% from the ranked list of outliers for all the datasets. ONE, though it is an unsupervised algorithm, outperforms all the  baseline algorithms in most of the cases. SEANO uses 20\% labeled data for training.}
  \label{fig:outdetec}
\end{figure*}


\subsection{Node Classification}
Node classification is an important application when labeling information is available only for a small subset of nodes in the network. This information can be used to enhance the accuracy of the label prediction task on the remaining/unlabeled nodes. 
For this task, firstly we get the embedding representations of the nodes and take them as the features to train a random forest classifier \cite{liaw2002classification}.
We split the set of nodes of the graph into training set and testing set. The training set size is varied from 10\% to 50\% of the entire data. The remaining (test) data is used to compare the performance of different algorithms.
We use two popular evaluation criteria based on F1-score, i.e., Macro-F1 and Micro-F1 to measure the performance of the multi-class classification algorithms. Micro-F1 is a weighted average of F1-score over all different class labels. Macro-F1 is an arithmetic average of F1-scores of all output class labels. Normally, the higher these values are, the better the classification performance is. We repeat each experiment 10 times and report the average results.

On all the datasets (Fig. \ref{fig:classi}), ONE consistently performs the best for classification, both in terms of macro and micro F1 scores. We see that conventional embedding algorithms like node2vec and TADW, which are generally good for consistent datasets, perform miserably in the presence of just 5\% outliers. AANE is the second best algorithm for classification in the presence of outliers, and it is very close to ONE in terms of F1 scores on the Citeseer dataset. {\it ONE is also able to outperform SEANO with a good margin, even though SEANO is a semi-supervised approach and uses labeled data for generating node embeddings}.

\begin{figure}[h!]
  \centering
  \begin{subfigure}[b]{0.49\linewidth}
    \includegraphics[width=\linewidth]{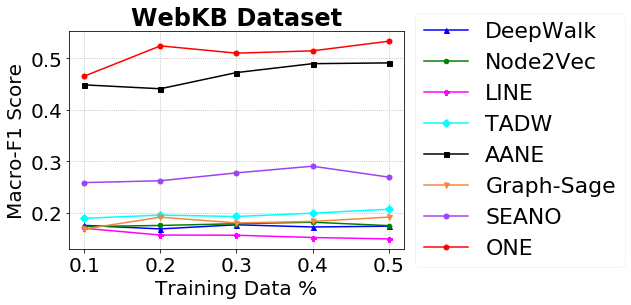}
  \end{subfigure}
  \begin{subfigure}[b]{0.49\linewidth}
    \includegraphics[width=\linewidth]{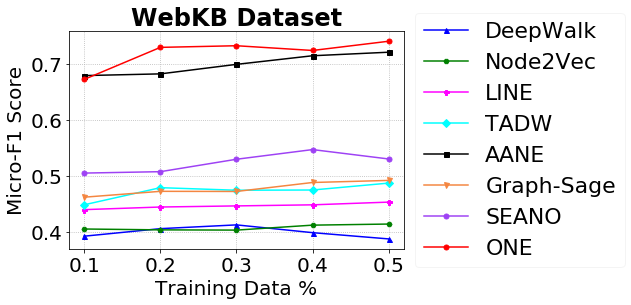}
  \end{subfigure}

  \centering
  \begin{subfigure}[b]{0.49\linewidth}
    \includegraphics[width=\linewidth]{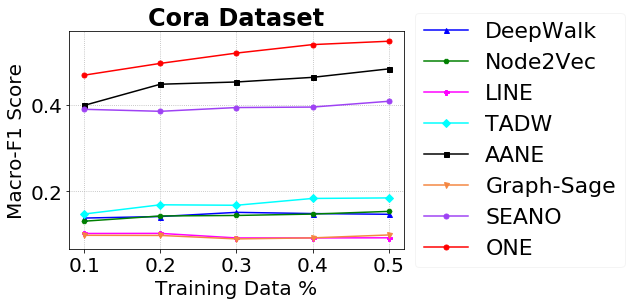}
  \end{subfigure}
  \begin{subfigure}[b]{0.49\linewidth}
    \includegraphics[width=\linewidth]{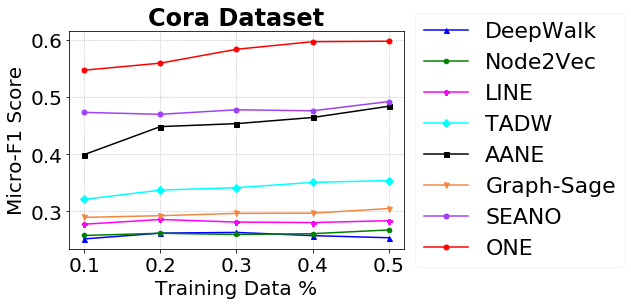}
  \end{subfigure}

  \centering
  \begin{subfigure}[b]{0.49\linewidth}
    \includegraphics[width=\linewidth]{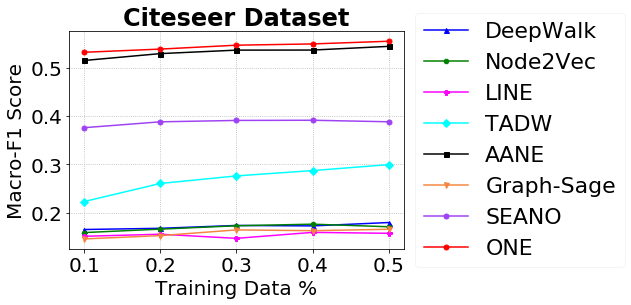}
  \end{subfigure}
  \begin{subfigure}[b]{0.49\linewidth}
    \includegraphics[width=\linewidth]{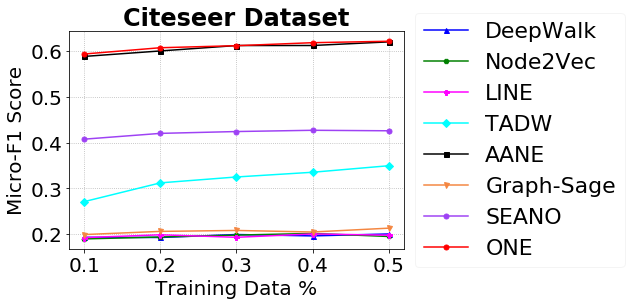}
  \end{subfigure}

  \centering
  \begin{subfigure}[b]{0.49\linewidth}
    \includegraphics[width=\linewidth]{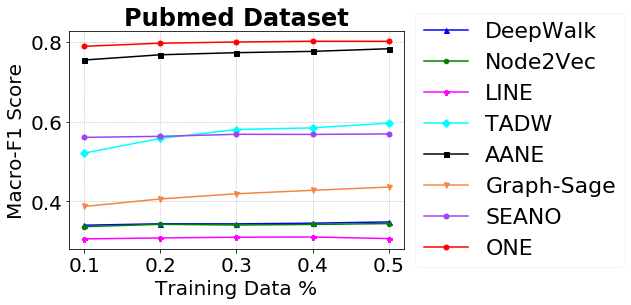}
  \end{subfigure}
  \begin{subfigure}[b]{0.49\linewidth}
    \includegraphics[width=\linewidth]{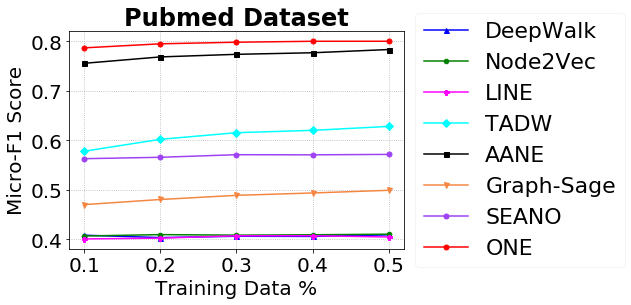}
  \end{subfigure}
   \caption{Performance of different embedding algorithms for Classification with Random Forest}
	\label{fig:classi}
\end{figure}

\subsection{Node Clustering}
Node Clustering is an unsupervised method of grouping the nodes into multiple communities or clusters.
First we run all the embedding algorithms to generate the embeddings of the nodes.
We use the node's embedding as the features for the node and then apply KMeans++ \cite{kmeans++}.
KMeans++ just divides the data into different groups. To find the test accuracy we need to assign the clusters with an appropriate label and compare with the ground truth community labels. For finding the test accuracy we use unsupervised clustering accuracy \cite{xie2016unsupervised} which uses different permutations of the labels and chooses the label ordering which gives best possible accuracy $Acc(\mathcal{\hat{C}},\mathcal{C}) = \max_{\mathcal{P}} \frac{ \sum\limits_{i=1}^N \mathbf{1}(\mathcal{P}(\mathcal{\hat{C}}_i)  = \mathcal{C}_i) }{N}$.
Here $\mathcal{C}$ is the ground truth labeling of the dataset such that $\mathcal{C}_i$ gives the ground truth label of $i$th data point. Similarly $\mathcal{\hat{C}}$ is the clustering assignments discovered by some algorithm, and $\mathcal{P}$ is a permutation of the set of labels.
We assume $\mathbf{1}$ to be a logical operator which returns 1 when the argument is true, otherwise returns 0.
Clustering performance is shown and explained in Fig. \ref{fig:clus}. It can be observed that except for ONE, all the conventional embedding algorithms fail in the presence of outliers. Our proposed unsupervised algorithm is able to outperform or remain close to SEANO, though SEANO is semi supervised in nature, on all the datasets.

\begin{figure}[h!]
  \centering
    \includegraphics[width=0.7\columnwidth]{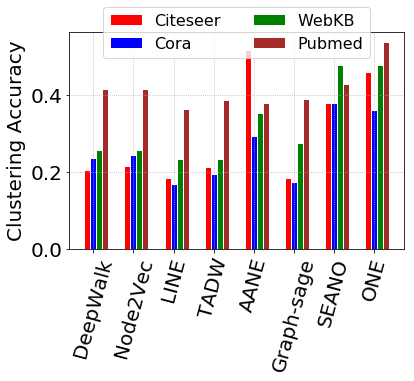}
  \caption{Clustering Accuracy of KMeans++ on the embeddings generated by different algorithms. ONE is always close to the best of the baseline algorithms. AANE works best for Citeseer. Though SEANO uses 20\% labeled data as the extra supervision to generate the embeddings, its accuracy is always close (or less) to ONE which is completely unsupervised in nature.}
  \label{fig:clus}
\end{figure}


\section{Discussion and Future Work}
We propose ONE, an unsupervised attributed network embedding approach that jointly learns and minimizes the effect of outliers in the network. We derive the details of the algorithm to optimize the associated cost function of ONE. Through experiments on seeded real world datasets, we show the superiority of ONE for outlier detection and other downstream network mining applications.

There are different ways to extend the proposed approach in the future. One interesting direction is to parallelize the algorithm and check its performance on real world large attributed networks. Most of the networks are very dynamic now-a-days. Even outliers also evolve over time. So bringing additional constraints in our framework to capture the dynamic temporal behavior of the outliers and other nodes of the network would also be interesting. As mentioned in Section \ref{sec:algOne}, ONE converges very fast on real datasets. But updating most of the variables in this framework takes $O(N)$ time, which leads to $O(N^2)$ runtime for ONE without any parallel processing. So, one can come up with some intelligent sampling or greedy method, perhaps based on random walks, to replace the expensive sums in various update rules.

{\footnotesize
\bibliography{ONE}
\bibliographystyle{aaai}}

\end{document}